\documentclass{llncs}
\usepackage{fullpage}
\usepackage{etoolbox}
\usepackage{epsfig}
\usepackage{framed}
\usepackage{color}
\usepackage{subfigure,soul}
\usepackage{graphicx}
\usepackage{comment}
\usepackage{ifthen}
\usepackage{amsmath,amssymb}
\usepackage{tikz}
\let\doendproof\endproof
\renewcommand\endproof{~\hfill\qed\doendproof}
\newcommand{\keywords}[1]{\par\addvspace\baselineskip
\noindent\keywordname\enspace\ignorespaces#1}

\newbool{fullversion}
\booltrue{fullversion}

\begin{document}

\title{Scheduling over Scenarios on Two Machines\thanks{This work was partially supported 
by EU-IRSES grant EUSACOU and Tinbergen Institute. EF was partially supported by projects PRH PICT 2009-119 
and UBACYT 20020090100149
. AvZ was partially supported by Suzann Wilson Matthews summer research award. }}
\author{Esteban Feuerstein\inst{1} \and Alberto Marchetti-Spaccamela\inst{2} \and Frans Schalekamp\inst{3} \and Ren\' e Sitters\inst{4,5} \and Suzanne van der Ster\inst{4} \and Leen Stougie\inst{4,5} \and Anke van Zuylen\inst{3}}
\institute{Departamento de Computaci\'on - FCEyN - UBA,
Buenos Aires, Argentina \\
\email{efeuerst@dc.uba.ar}
\and
Sapienza University of Rome, Italy\\
\email{alberto@dis.uniroma1.it}
\and
College of William and Mary, Department of Mathematics, Williamsburg VA 23185\\
\email{\{frans, anke\}@wm.edu}
\and
Vrije Universiteit Amsterdam, The Netherlands\\
\email{\{r.a.sitters, suzanne.vander.ster, l.stougie\}@vu.nl}
\and
CWI Amsterdam, The Netherlands\\
\email{\{r.a.sitters, stougie\}@cwi.nl}
}

\maketitle              

\begin{abstract}
We consider scheduling problems over scenarios where the goal is to find a single assignment of the jobs to the machines which performs well over all possible scenarios. Each scenario is a subset of jobs that must be executed in that scenario and all scenarios are given explicitly. The two objectives that we consider are minimizing the maximum makespan over all scenarios and minimizing the sum of the makespans of all scenarios. For both versions, we give several approximation algorithms and lower bounds on their approximability. With this research into optimization problems over scenarios, we have opened a new and rich field of interesting problems.
\end{abstract}

\section{Introduction}
We consider optimization problems {\it over scenarios}
where the goal is to find a single solution that performs well for each scenario in a given set of scenarios.
In particular, we consider the scheduling problem where the objective function is the makespan:
we are given a set $J$ of jobs, each with a processing time, and a set of scenarios; each scenario is specified by a subset of jobs in $J$ that must be executed in that scenario. Our goal is to find an assignment of jobs to machines that \emph{is the same for all scenarios} and optimizes a function of the \emph{makespan}, i.e., the completion time of the last completed job, \emph{over all scenarios}. The two objectives that we consider are minimizing the \emph{maximum} makespan over all scenarios and minimizing the \emph{sum} of the makespan of all scenarios.
We note that when the input contains only a single scenario, both versions of the problem reduce to the usual makespan minimization problem.

As an example, suppose that $J$ contains three jobs, numbered $1$, $2$, and $3$, that must be executed on two machines; the  processing time of job $1$ is $2$ while the processing time of jobs $2$ and $3$ is  $1$. There are three scenarios $S_1=\{1,2,3 \}$ and $S_2=S_3=\{2,3\}$. Assigning job $1$ to the first machine and jobs $2$ and $3$ to the second machine minimizes the maximum makespan over all scenarios, while assigning jobs $1$ and $2$ to the first machine and job $3$ to the second one minimizes the sum of the makespans of all scenarios.

The more egalitarian objective function of minimizing the maximum makespan over all scenarios fits in the framework of \emph{robust} optimization, where usually not so much a finite set of scenarios is explicitly given, as in our problem, but ranges for values of input parameters (see \cite{BenTalNemirovski}). We will refer to this objective as the MinMax objective. The more \emph{utilitarian} objective function of minimizing the sum of the makespans of all scenarios fits in the framework of \emph{a priori} optimization, though a priori optimization has so far only been introduced as a problem where the scenarios are random objects and the objective is to minimize the expected objective value. In that sense, minimizing the sum of makespans could be seen as the a priori problem with a uniform discrete distribution over a finite set of fully specified scenarios.  In general, the deterministic problem of optimizing over a finite set of scenarios can be seen as an alternative to the stochastic a priori setting \cite{Jaillet88}, in case a limited number of likely scenarios exists. We refer to this objective as the MinSum objective.

In an indirect way, combinatorial optimization problems over scenarios with the MinSum objective have appeared as the first-stage problem in a boosted sampling approach to two-stage stochastic optimization problems \cite{GuptaEA11}. In \cite{GuptaEA11}, scenarios are defined within a so-called black box, meaning that they can only be learnt by sampling.
From the black box, a finite set of scenarios is sampled, giving rise to a deterministic optimization problem over the drawn set of scenarios, in which a single solution needs to be found, that minimizes the sum of the objective values for the drawn scenarios.
In this sense some results on combinatorial optimization problems over scenarios have appeared, like {\sc Vertex Cover}, {\sc Steiner Tree} and {\sc Uncapacitated Facility Location} \cite{GuptaEA11}.

Modeling optimization problems over a finite set of given scenarios yields a rich source of interesting new combinatorial optimization problems, which are in general harder than their single-scenario versions. Specifically, almost any single-scenario scheduling problem has an interesting multi-scenario variation. As mentioned before, in this paper we focus, as a first example, on minimizing the maximum makespan over all scenarios and minimizing the sum (or, equivalently, the average) of the makespan of all scenarios.

The specific setting of the scheduling problem over scenarios  appears in situations where jobs have to be performed by skilled machines (workers), and some investment is required to attain the skill for a particular job. In such situations, one should decide on an assignment of all possible jobs to the workers, such that the workers can train for the jobs assigned to them ahead of time. The problem then is to assign jobs (specializations) to machines (workers), so that the workload of a machine for any scenario of jobs, from a set of scenarios likely to occur, is minimized.
Examples of such a setting are assignment of clients to lawyers, 
households to power sources, compile-time assignment of computational tasks to processors.
In most of such situations, the robust version of the problem with the \mbox{MinMax} objective is rather plausible, especially in situations where a set of likely scenarios to hedge against can be specified upfront.

Another motivation, though a bit indirect, comes from distributed information retrieval: in a term-partitioned index, it is good to allocate to the same processor terms appearing frequently together in queries, so as to minimize the communication cost (to solve an intersection query between two terms that reside in different processors, one of the posting lists must be sent to the processor holding the other). But this goal must be complemented with that of balancing the load, as it is not viable to put all the terms in the same processor. Therefore, it is necessary to divide ``clusters'' of commonly co-occurring terms among the processors, trying to balance the load. Naturally, queries appear sequentially over time and are not known a priori. One could, as an approximation, optimize considering as input the more likely scenarios. The partition must indeed be done a priori, because lists must be assigned to processors a priori.

To the best of our knowledge, this problem has not been considered in the literature. An a priori version of scheduling with stochastic scenarios has been studied in \cite{Bouyahia2013,Bouyahia2010488}, albeit not from an approximation theory point of view, but merely presenting experimental results, and with the scheduling objective of minimizing the sum of completion times of all the jobs per scenario.

We now give a formal definition of the two problems we consider.
We restrict ourselves to the case of two machines.
We are given a set of jobs $J$ with for each job $j\in J$ a processing times $p_j$, and a set of $k$ scenarios ${\cal S}=\{S_1,S_2,\ldots,S_k\}$, where each scenario $S_i\in {\cal S}$ is a subset of $J$. In each scenario,
we are interested in minimizing the makespan, but we are restricted to finding a solution, i.e., an assignment of the jobs to the machines, that applies to every one of the scenarios. Clearly, a solution that is good for one scenario may be bad for another. This gives rise to specifying objectives that reflect the trade-off between the various scenarios. In this paper we define the following two versions of the problem.
\begin{itemize}
\item {\bf MM2}
Assign the jobs in $J$ to two machines in such a way that the maximum makespan over the given scenarios is minimized. In other words, if we denote the makespan of a subset $S\subseteq J$ of jobs by $p(S)= \sum_{j\in S}p_j$, we are looking for a partition $A,\bar A$ of $J$, that minimizes $\max_{i=1,2,\ldots,k} \max\{ p( A \cap S_i ), p( \bar A \cap S_i )\}$. \\
\item {\bf SM2}
Assign the jobs to the machines such that the sum of the makespans of the given scenarios is minimized.  Using the notation just introduced, we are seeking a partition $A,\bar A$ of $J$, that minimizes $\sum_{i=1}^k \max\{ p( A \cap S_i ), p( \bar A \cap S_i )\}$.
\end{itemize}
For both objective functions, the problems are NP-hard, since the single-scenario version is NP-hard. However, the single-scenario version is only weakly NP-hard for 2 machines and an FPTAS exists \cite{ausiello}, whereas the problems defined here are strongly NP-hard. We will give various approximability and inapproximability results for several different versions of the problem depending on restrictions of the input.
In particular, the special cases that we consider are the following:
\begin{enumerate}
\item
$p_j=1$ $\forall j\in J$, that is, the case where all processing times are unitary;
\item
$|S_i | \leq r$ $\forall S_i \in {\cal S}$, that is, the case where the number of jobs in each scenario is bounded by a constant;
\item
$k = |{\cal S}|$ is constant, that is, the case that the number of scenarios is a constant.
\end{enumerate}

In Section~\ref{sec:minmax}, we study the problem MM2; we show that the problem cannot be approximated to within a ratio of $2- \varepsilon$ already in the case where $p_j=1$   and a ratio of $3/2$ if  $|S_i| \le 3$ and $p_j=1$. On the positive side, we give a  polynomial-time algorithm for the version in which every scenario contains 2 jobs. If $k$, the number of scenarios, is constant then there exists a PTAS; for an arbitrary number of scenarios, a $O(\log ^2k)$ approximation ratio exists. The latter two results are a consequence of an observed direct relation to the so-called {\sc Vector Scheduling} problem (see Section~\ref{sec:minmax} for its definition) and  results of~\cite{chekuri2004}.

In Section~\ref{sec:minsum}, we study problem SM2. We prove inapproximability within $1.0196$ assuming P$\neq$NP, and within $1.0404$ under the \emph{Unique Games Conjecture}~\cite{khotUGC}.
On the positive side, we present a $3/2$-approximate randomized algorithm. For instances with scenarios of size at most 3, we use a reduction to {\sc Max Cut} to obtain a 1.12144-approximation algorithm. For scenarios of size at most $r$, we present a reduction to {\sc Weighted Max Not-All-Equal $r$-Sat} and use this to obtain better-than-3/2- approximations for problems where the scenario sizes are not larger than 4.

Some thoughts about related problems, and ideas for future research are contained in a concluding section.

\section{Minimizing Maximum Makespan}
\label{sec:minmax}

We obtain inapproximability of MM2 using a recent result \cite{AustrinEA13} on the hardness of {\sc Hypergraph Balancing}: given a hypergraph find a 2-coloring of the vertices such as to minimize over all hyperedges the discrepancy between the number of vertices of the two colors.
\begin{theorem}
It is NP-hard to approximate MM2 with unitary jobs within ratio $2-\epsilon$.
\end{theorem}
This is  remarkable since, trivially, any solution, for any job sizes, is $2$-approximate (since we consider the problem for two machines only). In
\ifboolexpr{bool{fullversion}}
{the appendix we prove Theorem 9 that gives   a $3/2$ hardness bound}
{the full version of this paper we prove a hardness bound of $3/2$} 
when $|S_i| \le 3$ and $p_j=1$. This result completes the hardness characterization.

We now show that, if the number of jobs per scenario is 2, then the problem is solvable in polynomial time.

\begin{theorem}
MM2 with $|S_i| =2$ for all $S_i\in {\cal S}$ can be solved in time $O(|{\cal S}| \log |{\cal S}|)$.
\end{theorem}
\begin{proof}
We create a graph with a vertex for each job and connect by an edge the jobs that appear together in a scenario. We define the weight of edge $(j,k)$ to be
$p_j + p_k$, i.e., the sum of the processing times of the jobs associated to the incident vertices. Note that a solution for the a priori scheduling problem is a partitioning of the job set, and can be associated with a coloring of the vertices in this graph problem with two color classes. The objective value is then equal to the maximum of the highest weight of any monochromatic edge and the largest processing time of any job.

In other words, we should find a 2-coloring of the vertices of this graph, such that the maximum weight of a monochromatic  edge is minimized. A lowest weight edge in any odd cycle gives a lower bound on the objective value.

Consider the following algorithm. Starting with all vertices being part of their own singleton component, and having color $1$, we grow components by inserting edges, and label the vertices with the component they belong to, and with a color that can assume two values; $1$ and $2$. A color inversion of a vertex changes the color of the vertex (i.e., if it is colored $1$, the color is changed to $2$, and vice versa). We consider the edges in order of descending weight. When considering the next edge, say $(j,k)$, the following 3 cases can occur.

\smallskip
\noindent {\em Case 1.} Vertices $j$ and $k$ have the same color, and are in the same component. We end the algorithm. An optimal partitioning of the job set is given by the two color classes, where jobs that have color $1$ (respectively, $2$) are assigned to machine $1$ (respectively, $2$) and the objective is equal to the weight of edge $(j,k)$.

\smallskip

\noindent {\em Case 2.} Vertices $j$ and $k$ have different colors. If the vertices are in different components, then we update the component label for all nodes of the smaller component (breaking ties arbitrarily), so that all vertices have the same label. We then proceed to the next edge.

\smallskip

\noindent {\em Case 3.} Vertices $j$ and $k$ have the same color, and are in different components. In this case we invert the color of all nodes in the smaller component  (breaking ties arbitrarily), and then proceed as in Case 2.

\medskip
By construction, two vertices of the same color in the same component are joined by an even-length path. Therefore, when the algorithm terminates in Case 1, we have found an odd cycle in the graph, of which this last edge has lowest weight. Note that the assignment of jobs with the same color to the same machine implies that the makespan of the scenario is bounded by the weight of the last considered edge.
Since the weight of any such edge is a lower bound on the objective value, we have found an optimal solution.  Its value is given by the maximum of the weight of a monochromatic edge and $p_{\max}=\max_{j\in J} p_j$.

The running time follows from the observation that any time we invert the color and/or update the label of a vertex, it ends up in a component of at least twice the size of the component it belonged to before. Hence, the label of a vertex can be updated at most $\log |J|$ times. The total time can thus be bounded by $|{\cal S}|\log |{\cal S}|$ time for sorting the edges by weight, plus $|J|\log|J|$ time for updating the vertex colors and labels. Finally, we may assume without loss of generality that each job appears in at least one scenario, so $|{\cal S}|\ge |J|/2$.
\end{proof}

Another sharp characterization w.r.t. the number of scenarios, is obtained for the case of a constant number of scenarios. For jobs with unit processing times, the problem can be solved exactly: given that the number of scenarios is constant, there is only a constant number of job types, where the type of a job is the set of scenarios it is in. Then, the number of jobs on machine 1 of each type can be guessed. There are only a polynomial number of choices;, an extension can also accommodate a constant number of machines in polynomial time.
We notice that this also solves SM2 under the same restrictions in polynomial time. 

\begin{theorem}
\label{thm:dp}
MM2 and SM2 having jobs with unitary processing times can be solved in polynomial time if the number of scenarios is constant.
\end{theorem}

A similar idea, with guessing the optimal value and rounding, leads to a PTAS in the general case under a constant number of scenarios, but this is also implied by the following result.

We conclude this section by noticing that if we consider any number of machines, the problem of minimizing the maximum makespan reduces to the {\sc Vector Scheduling} problem, where each coordinate corresponds to a scenario.
\begin{definition}
In the {\sc Vector Scheduling} problem we are given a set $V$ of $n$ rational $d$-dimensional vectors $v_1, \ldots, v_n$ from $[0,\infty)^d$ and a number $m$. A valid solution is a partition of $V$ into $m$ sets $A_1, \ldots, A_m$. The objective is to minimize $\max_{1\leq i\leq m} ||\sum_{v_j\in A_i} v_j||_{\infty}$.
\end{definition}
This problem is a $d$-dimensional generalization of the makespan minimization problem, where each job is a $d$-dimensional vector and the machines are $d$-dimensional objects as well. In our setting, the dimension $d$ equals the number of scenarios $|{\cal S}|$. Each coordinate of job $j$ equals its processing time in the corresponding scenario (either 0 or $p_j$).  Results of Chekuri et al.~\cite{chekuri2004} on {\sc Vector Scheduling} can directly be translated into our setting.

\begin{theorem}[\cite{chekuri2004}] For the problem of minimizing the maximum makespan over scenarios $S_i\in \cal S$ on $m$ machines,
\begin{enumerate}
\item there exists a PTAS for the case that $k = |{\cal S}|$ is constant
\item there exists a polynomial-time $O(\log^2 k)$-approximation for  $k$ scenarios;
\item there exists no $c$-approximation algorithm for any $c>1$, when dealing with any number of scenarios.
\end{enumerate}
\end{theorem}

\section{Minimizing Sum of Makespans}
\label{sec:minsum}

\newcommand{\half}{\ensuremath{\tfrac12}}
\newcommand{\oneandhalf}{\ensuremath{3/2}}

We now turn our attention to SM2, the problem of minimizing the  sum of the makespans over all scenarios, in the case of 2 machines.

We start this section by noting that SM2 is MAX\,SNP-hard even with unitary processing times and scenarios containing two jobs each.

\begin{theorem}
SM2 is NP-hard to approximate to within a factor of $1.0196$ and UGC-hard to approximation to within a
factor of $1.0404$, even if all jobs have length 1, and all scenarios contain two jobs.
\end{theorem}

The proof is through a reduction from \textsc{Max Cut}~\cite{GareyJohnson}, and the hardness of approximation results shown by H{\aa}stadt~\cite{Hastad01} and Khot et al.~\cite{khot}. The details are given in the 
\ifboolexpr{bool{fullversion}}
{appendix.}
{full version of this paper.}

In the remainder of this section, we will give approximation results for SM2.  As for MM2 in the previous section, we notice that also for this problem any solution is a trivial 2-approximation.
In the remainder of this section, we will first show that the algorithm that randomly assigns the jobs to the two machines independently with equal probability gives a $3/2$-approximation. We then show two deterministic approximation algorithms, which give good approximation guarantees if the number of jobs per scenario is small.

\subsection{A Randomized Approximation Algorithm}
\begin{lemma}\label{lemma:half}
Consider a scenario $S$, and let $A,\bar A$ be any partitioning of the jobs in $S$.
When assigning each job of $S$ to the two machines independently with equal probability, the expected load of the least loaded machine is at least \half$\min\{p(A),p(\bar A)\}$.
\end{lemma}
\begin{proof}
An assignment of jobs to the two machines induces a partition of $A$ into sets $A', A''$, and a partition of $\bar A$ into sets $\bar A', \bar A''$ where the jobs in the same set of the partition are assigned to the same machine.
The sets $A', A'', \bar A', \bar A''$ are not necessarily all non-empty.
We will prove the lemma by showing that, conditioned on the sets $A', A'', \bar A', \bar A''$, the machine load of the least loaded machine is at least $\tfrac12 \min\{p(A),p(\bar A)\}$, which implies that the statement also holds unconditionally.

Conditioned on the sets $A', A'', \bar A', \bar A''$, the least loaded machine has a load of
$\min\{p(A')+p(\bar A'), p(A'')+p(\bar A'')\}$ with probability $\tfrac12$ (namely, if $A', \bar A'$ are assigned to one machine, and $A'',\bar A''$ to the other machine), and $\min\{p(A')+p(\bar A''), p(A'')+p(\bar A')\}$  with probability $\tfrac12$ (namely, if $A', \bar A''$ are assigned to one machine, and $A'', \bar A'$ are assigned to the other machine).
Hence, conditioned on the partition of $A$ into $A', A''$ and of $\bar A$ into $\bar A', \bar A''$, the expected load of the least loaded machine is
$$\tfrac12\min\{p(A')+p(\bar A'), p(A'')+p(\bar A'')\} + \tfrac12\min\{p(A')+p(\bar A''), p(A'')+p(\bar A')\}.$$

Note that a simple case analysis shows that the sum of the two terms is either at least
$\tfrac12 \left(p(A')+p(A'')\right)=\tfrac12p(A)$ or at least $\tfrac12\left(p(\bar A')+p(\bar A'')\right)=\tfrac12p(\bar A)$.
So the load of the least loaded machine is at least $\tfrac12\min\{p(A),p(\bar A)\}$.
\end{proof}

\begin{theorem}
\label{coro:random_2_machines}
Randomly assigning each job to the two machines independently with equal probability is a \oneandhalf-approximation for SM2.
\end{theorem}
\begin{proof}
Consider a scenario $S$, and let $A$ be the set of jobs processed on machine 1, and $\bar A=S\backslash A$ the set of jobs processed on machine 2 in a schedule with minimum makespan. Hence, the optimal makespan for $S$ is $\max\{p(A), p(\bar A)\}$.
By Lemma~\ref{lemma:half}, the load of the least loaded machine in scenario $S$, if the jobs are randomly assigned to the machines with equal probability, is at least $\tfrac12\min\{p(A), p(\bar A)\}$.
Hence, the load of the machine with the highest load is at most $p(A)+p(\bar A) - \tfrac12 \min\{p(A), p(\bar A)\} = \max\{p(A), p(\bar A)\} + \tfrac12 \min \{p(A), p(\bar A)\}\le \tfrac 32\max\{p(A), p(\bar A)\}$.

Hence, the expected makespan for scenario $S$ is at most $\tfrac32$ times the optimal makespan for scenario $S$, which implies that the sum over all scenarios of the expected makespans is at most $\tfrac32$ times the optimal summed makespan of all scenarios.
\end{proof}

We remark that the proof of the previous lemma bounds the objective value by comparing the load on a machine in a given scenario to the load for the optimal schedule {\it for that scenario}, rather than the optimal schedule for our problem.

It is easy to see that the analysis of the simple randomized algorithm is tight, by considering an instance of two jobs $\{1,2\}$ with unitary execution time and one scenario $S_1=\{1,2\}$. The optimal solution is to assign one job to each machine, whereas the randomized algorithm either assigns both jobs to the same machine with probability $\half$, or one job to each machine with probability $\half$.

\subsection{Deterministic Approximation Algorithms}\label{sec:nae_sat_reduction}
To obtain a deterministic approximation algorithm, we show that the SM2 problem can be reduced to the {\sc Weighted Max Not-All-Equal Satisfiability} problem, that we will abbreviate as {\sc Max-Nae Sat}.

\begin{definition}
In {\sc Max-Nae Sat}, a boolean expression is given, and a weight for each clause. A clause in the expression is satisfied if it contains both true and false literals. The problem is to find an assignment of true/false values to the variables, such as to maximize the total weight of the clauses satisfied. \end{definition}

Note that if $r$ is such that $|S_i|\leq r$ for all $S_i\in {\cal S}$, then by adding dummy jobs of processing time 0, we can assume that every scenario contains exactly the same number of jobs, i.e., $|S_i| =r$ for all $S_i\in {\cal S}$.
We will reduce the SM2 problem with scenarios of size at most $r$ to the {\sc Max-Nae Sat} problem with clauses of length $r$ ({\sc Max-Nae $r$-Sat}).

\begin{theorem}
A $(1-\gamma_r)$-approximation for {\sc Max-Nae $r$-Sat} implies a $(1+2^{r-2}\gamma_r)$-approximation for the SM2 problem with $|S| \le r$ for all scenarios $S \in {\cal S}$.
\end{theorem}
\begin{proof}
We start by formulating the SM2 problem as a {\sc Max-Nae Sat} problem. Each job $j$ corresponds to a variable $x_j$ in the {\sc Max-Nae Sat} instance.
An assignment of the variables in the {\sc Max-Nae Sat} instance corresponds to an assignment in SM2 as follows: machine 1 is assigned all jobs for which the corresponding variable is set to true, and machine 2 processes all jobs for which the corresponding variable is set to false.

We now construct a set of weighted clauses for each scenario such that the weight of the satisfied clauses for a given assigment is equal to the load of the least loaded machine in the scenario. Hence, maximizing the weight of the satisfied clauses will maximize the weight of the least loaded machine, and it will thus minimize the weight of the machine with the heaviest load, i.e., the makespan.

For a given scenario $S$ of SM2 with $r$ jobs, we construct $2^{r-1}$ clauses of length $r$ as follows.
For each partitioning of $S$ into two  sets $A$ and $\bar A$, we create a clause denoted by $C_S( \{A,\bar A\} )$. In clause $C_S( \{A,\bar A\} )$, all variables corresponding to jobs in one set appear negated, all variables corresponding to the other set appear non-negated. Note that $C_S(\{A,\bar A\})$ has the same truth table as $C_S(\{\bar A, A\})$ (namely, a clause is false if and only if all its literals are false, or all its literals are true).
Note that this means that if $A$ is assigned to the first machine and $\bar A$ is assigned to the second machine, then all clauses {\it except} $C_S( \{A, \bar A\} )$ are satisfied.

Denote by $w_S( \{ A, \bar A \} )$ the weight on the clause $C_S( \{ A, \bar A \} )$. To ensure the weight of the satisfied clauses is equal to the weight of the least loaded machine in SM2, we define weights on the clauses to be so that
\begin{equation*}
\sum_{{B, \bar B: \\ B \cup \bar B = S,  B \cap \bar B = \emptyset}} w_S( \{B, \bar B\} ) - w_S( \{ A, \bar A \} ) = \min\{ p( A ), p(\bar  A) \}.
\end{equation*}
Let $N = 2^{r-1}$, i.e., $N$ is the number of clauses corresponding to scenario $S$. The solution to this system of equations is to set
\[
w_S(\{A, \bar A\}) = \frac 1 {N-1} \sum_{{B, \bar B: \\ B\cup \bar B= S,  B \cap \bar B= \emptyset}}\min\{p(B), p(\bar B)\} - \min\{p(A), p(\bar A)\}.\]
The weights thus defined are not necessarily non-negative: consider a scenario $S$ that contains $r=4$ jobs of unit length. There are four ways of partitioning $S$ into one set of size one and one set of size three, and there are ${4 \choose 2} / 2 = 3$ ways of partitioning $S$ into two sets of size two.  Therefore $\sum_{\substack{B, \bar B}}
\min\{p(B), p(\bar B)\}= 10$, but that means that for a partitioning into sets $A, \bar A$ of size two $w_S(\{A, \bar A\}) = \tfrac 17 (10) - 2<0$.





To use approximation algorithms for {\sc Max-Nae Sat}, we need to make sure that all weights are non-negative. We accomplish this by adding a constant $K(S)$ to all weights of clauses corresponding to scenario $S$, where we set $-K(S)$ equal to a lower bound on the weights. We derive a lower bound on the weights by noting that (1) $\tfrac 1 {N}  \sum_{ B, \bar B }\min\{p(B), p(\bar B)\}$ is the expected value of the least loaded machine when all jobs are assigned to a machine with probability \half\ independently, hence, by Lemma~\ref{lemma:half}, its value is lower bounded by $\tfrac 12 \max_{B,\bar B} \min\{ p( B ), p( \bar B ) \}$; and (2) trivially,  $\max_{B, \bar B} \min\{ p(B), p(\bar B) \} \le \tfrac 12 p(S)$. Therefore
\begin{eqnarray*}
	w_S(\{A, \bar A\}) &=& 
	\tfrac 1 {N-1} \sum_{{B, \bar B: B \cup \bar B = S, B \cap \bar B = \emptyset} }\min\{p(B), p(\bar B)\} - \min\{p(A), p(\bar A)\}  \\
	&=& \tfrac {N}{N-1} \tfrac 1N \sum_{{B, \bar B: B \cup \bar B = S,  B \cap \bar B = \emptyset} }\min\{p(B), p(\bar B)\} - \min\{p(A), p(\bar A)\}  \\	
	&\geq& \tfrac {N}{N-1}\tfrac 12 \max_{B, \bar B} \min\{ p( B ), p(\bar B ) \} - \min\{p(A), p(\bar A)\}  \\
	&\geq& \tfrac {\tfrac 12 N - (N - 1) }{N-1 } \max_{B, \bar B} \min\{ p( B ), p( \bar B ) \} \\
	&=& -\tfrac 12 \tfrac {N- 2}{N-1 } \max_{B,\bar B} \min\{ p( B ), p( \bar B) \} \\
	&\geq& -\tfrac 14 \tfrac {N-2}{N-1} p(S).
\end{eqnarray*}
Thus, we set $K( S ) = \tfrac 14 \tfrac {N-2}{N-1} p( S )$, such that $\tilde w_S(\{A, \bar A\}) =  w_S(\{A, \bar A\}) + K( S ) \geq 0$ for all partitionings $A, \bar A$ of $S$ into two sets.





A solution to the {\sc Max-Nae Sat} instance is now mapped to a solution of SM2, by assigning the jobs for which the variable is set to true to machine 1, and scheduling the other jobs on machine 2.
We note that the $w$-weights of the clauses corresponding to scenario $S$ were chosen so that the sum of the weights of the clauses that are satisfied is exactly equal to the load on the least loaded machine in scenario $S$. Also,
$N-1$ clauses of scenario $S$ are satisfied in any solution to the {\sc Max-Nae Sat} instance. Therefore the total $\tilde w$-weight of the clauses for scenario $S$ that are satisfied in any {\sc Max-Nae Sat} solution is equal to the load on the least loaded machine in scenario $S$ plus an additional $(N-1) K(S)$.

We let $L = \sum_S p( S )$, and denote by $L^*_{\min}$ the sum over all scenarios of the load of the least loaded machine in an optimal solution, and by $L^*_{\max}$ the sum over all scenarios of the load of the most loaded machine in an optimal solution, so that $L^*_{\min} + L^*_{\max} = L$. Note that the additional term $K(S)$ in the $\tilde w$-weights of the {\sc Max-Nae Sat} solution causes an increase of the objective value with respect to the $w$-weights solution by adding an additional $\sum_S (N-1) K(S) = \sum_S \tfrac 14 (N-2) p( S ) = \tfrac 14 (N-2) L$ to each solution.

In particular, an optimal solution to the {\sc Max-Nae Sat} instance, has objective value $L^*_{\min} + \tfrac 14 (N-2) L$, and a $(1-\gamma)$-approximation algorithm for the {\sc Max-Nae Sat} instance, therefore, has objective value at least $(1-\gamma) ( L^*_{\min} + \tfrac 14 (N-2) L )$. Let us denote by $ALG(L_{\min})$ and $ALG(L_{\max})$ the sum over all scenarios of the least and most loaded machines in the corresponding job assignment. Note that
$ALG(L_{\min}) \geq (1-\gamma)\left(( L^*_{\min} + \tfrac 14 (N-2) L \right) - \tfrac 14 (N-2) L 
= (1-\gamma) L^*_{\min} - \tfrac 14\gamma (N-2) L$. 
Therefore,
\begin{align*}
ALG( L_{\max} ) = L - ALG(L_{\min}) &\leq L - ((1-\gamma) L^*_{\min} - \tfrac 14\gamma (N-2) L)\\
&= (1-\gamma) ( L - L^*_{\min} ) + \gamma L + \tfrac 14\gamma (N-2) L \\
&= (1-\gamma) L^*_{\max} + \tfrac 14\gamma (N+2) L.
\end{align*}
Noting that $L \leq 2L^*_{\max}$ gives  $ ALG( L_{\max} ) \leq (1-\gamma) L^*_{\max} + \tfrac 12\gamma (N+2) L^*_{\max} = ( 1 + \tfrac 12 \gamma N ) L^*_{\max}$
which proves the theorem, since $N=2^{r-1}$.
\end{proof}

For $r=3$, Zwick~\cite{Zwick99} gives a 0.90871-approximation for {\sc Max-Nae 3-Sat}. By the previous lemma, this gives a 1.18258-approximation for SM2 with scenarios of length at most three. For $r=4$, Karloff et al.~\cite{KarloffZ97} give a $\tfrac78$-approximation for {\sc Max-NAE 4-Sat}. By our lemma, this implies a $\tfrac32$-approximation for SM2 with scenarios of size 4. Note that this matches the guarantee we proved for the algorithm that randomly assigns each job to one of the two machines.
For general $r$, the best approximation factor known for {\sc Max-Nae Sat} is 0.74996 due to Zhang, et al.~\cite{ZhangYH04}, and the implied approximation guarantees for our problem are worse than the guarantee for the random assignment.

If every scenario has exactly two jobs, then we can obtain a better approximation guarantee by reducing SM2 to  {\sc Max Cut} as follows: we create a vertex for every job, and add an edge between $i$ and $j$ of weight $\min\{p_i, p_j\}$ for every scenario that contains jobs $i$ and $j$. For any cut, the weight of the edges crossing the cut is then exactly the sum over all scenarios of the load of the least loaded machine. Since the makespan for a scenario $S$ is $p(S)$ minus the load of the least loaded machine, maximizing the load of the least loaded machine, summed over all scenarios, is equivalent to minimizing the sum of the makespans.

If every scenario has at most three jobs, we can also reduce SM2  to {\sc Max Cut}, but the reduction, given in the 
\ifboolexpr{bool{fullversion}}
{appendix,}
{full version,}
is slightly more involved.

\begin{theorem}
\label{thm:UB_SM2_MAXCUT}
There exists a $(1+\gamma)$-approximation algorithm for the SM2 problem with scenarios containing at most three jobs, where $1-\gamma$ is equal to the approximation ratio for {\sc Max Cut}.
\end{theorem}

The 0.87856-approximation for {\sc Max Cut} of Goemans et al.~\cite{GoemansW95} gives us the following corollary.

\begin{corollary}
\label{coro:UB_SM2_MAXCUT}
There exists a $1.12144$-approximation algorithm for the SM2 problem with scenarios containing at most three jobs. \end{corollary}

\section{Epilogue}

This paper presents some first results on a basic scheduling problem under a set of scenarios.
The objective is to find a single solution that is applied to all the scenarios specified.
We studied this problem for scheduling with two different objectives: minimizing the maximum objective value over all scenarios, the MinMax version, and minimizing the sum of the objective values of all scenarios, the MinSum version.

To the best of our knowledge, combinatorial optimization problems under a set of fully explicitly specified scenarios has hardly been studied in the literature. Apart from posing theoretically interesting questions as we hope to have shown with this paper, it enhances our ability to model decisions problems where a learning aspect for performing jobs prohibits that job assignments can be adjusted on a day-by-day basis, but merely require a fixed assignment whose quality then necessarily differs over the various instances.

In relation to the MinMax version of the problem, we also like to mention a version of combinatorial optimization which has become known under the name universal optimization. E.g., \cite{EpsteinEA12} study a universal scheduling problem. In such a problem, the scenarios are not explicitly specified, but can be seen to be chosen by an adversary. The quality of an algorithm is then measured by comparing its solution to the optimal solution when the adversarial choices are known beforehand.

For future research, anyone can choose her or his favorite combinatorial optimization problem and study its multiple-scenario version.

We finish with the some questions emerging from our multiple-scenario scheduling problem. The result in \cite{AustrinEA13} suggests a 3/2-approximation for MM2 with 4 jobs per scenario and unitary jobs. Can this be extended to any job sizes? For the SM2 version the question is to close the gap between the 3/2-approximate randomized algorithm for the general case and the 1.0404 lower bound under the Unique Games Conjecture. It would also be interesting to find out if our randomized algorithm can be derandomized.

\newpage

\bibliographystyle{plain}
\bibliography{a_priori_scheduling}

\appendix
\section{Deferred Proofs}

\setcounter{theorem}{8}
\begin{theorem}
It is NP-hard to approximate MM2 with scenarios of size 3 to within a factor of $3/2$.
\end{theorem}
\begin{proof}
In~\cite{GareyJohnson}, it is shown that {\sc Set Splitting} is NP-complete, even if all sets have size 3. Consider such an instance of {\sc Set Splitting}. For each object in the {\sc Set Splitting} instance, we introduce a job with processing time 1. For each set in the input we introduce a scenario with the corresponding three jobs. Then, there is a partition of the jobs into two sets $A$ and $\bar A$ such that the objective value of MM2 is 2 if and only if there is a YES-answer to the {\sc Set Splitting} instance.

In particular, the NP-completeness with sets of size 3 implies a $3/2$ inapproximability result. The makespan for every scenario is either 2 or 3. Approximating this problem to within a factor better than $3/2$ implies being able to solve the decision problem {\sc Set Splitting}.
\end{proof}

\setcounter{theorem}{4}
\begin{theorem}
SM2 is NP-hard to approximate to within a factor of $1.0196$ and UGC-hard to approximation to within a
factor of $1.0404$, even if all jobs have length 1, and all scenarios contain two jobs.
\end{theorem}

\begin{proof}
We use a reduction from \textsc{Max Cut}~\cite{GareyJohnson}.
Given an (unweighted) max cut instance $G$, we create a job with processing time 1 for each vertex, and for each edge we create a scenario with two jobs that correspond to the vertices incident to that edge.
A cut in $G$ is induced by a partition of the vertices, which correspond to a partition $A, \bar A$ of the jobs. Scenario $(j,h)$ contributes 1 to the objective of the scheduling problem if and only if $|\{ j,h\} \cap A|=1 $, i.e., edge $(j,h)$ is in the cut, and it contributes 2 otherwise (because if the two jobs are assigned to the same machine, the makespan of that scenario is 2). Thus, the objective value of the scheduling problem is equal to twice the total number of edges in the graph minus the size of the cut.

Now, let $OPT(CUT)$ be the optimal {\sc Max Cut} objective, let $m$ be the number of edges in the instance, and assume we have an $(1+ \alpha)$-approximation algorithm for our problem, i.e., a solution with sum of makespans at most $(1+\alpha)(2m-OPT(CUT))$.
Hence, at least $(1+\alpha) OPT(CUT)-\alpha 2m$ scenarios have a makespan of 1, i.e., the size of the corresponding cut in $G$ is at least $(1+\alpha) OPT(CUT)-\alpha 2m$. Now, note that $OPT(CUT)\ge m/2$, and hence the size of the cut is at least $(1-3\alpha)OPT(CUT)$.

By the lower bound on the approximability of {\sc Max Cut} proved by H{\aa}stadt~\cite{Hastad01}, $1-3\alpha \ge 0.941176$ unless $P=NP$,  and by a result of Khot et al.~\cite{khot}, $1-3\alpha\ge 0.878567$ under the Unique Games Conjecture.
\end{proof}

\setcounter{theorem}{7}
\begin{theorem}
\label{thm:UB_SM2_MAXCUT}
There exists a $(1+\gamma)$-approximation algorithm for the SM2 problem with scenarios containing at most three jobs, where $1-\gamma$ is equal to the approximation ratio for {\sc Max Cut}.
\end{theorem}

\begin{proof}

We may assume by adding dummy jobs with processing time 0 that every scenario has three jobs.
We create a vertex for every job, and for a scenario containing jobs $i$, $j$ and $k$, we add edges $\{i,j\}$, $\{j,k\}$ and $\{k,i\}$. 
If multiple scenarios contain jobs $i$ and $j$, the corresponding edge will have the same multiplicity in the constructed graph.
Note that a cut will either have zero or two of the edges corresponding to a given scenario crossing the cut. We now set the weight of the edges in such a way that if two edges cross the cut, then the sum of the weights of the two edges is equal to the load of the least loaded machine.
In order to do this, we first define $b_i$ to be the load of the least loaded machine in the scenario, if $i$ is on one machine, and $j$ and $k$ are on the other machine, i.e., $b_i = \min\{p_i, p_j+p_k\}$. We similarly define $b_j=\min\{p_j, p_i+p_k\}$ and $b_k=\min\{p_k, p_i+p_j\}$.
Then we want to set the weights $w(e)$ such that
\begin{eqnarray*}
w(i,j)+w(i,k)&=& b_i;\\
w(i,j)+w(j,k)&=& b_j;\\
w(i,k)+w(j,k)&=& b_k.
\end{eqnarray*}
This is a system of three linearly independent equations with three unknowns, which has the (unique) solution
$w(e)=\tfrac12(b_i+b_j+b_j)-b_v$, where $e\in \{\{i,j\}, \{j,k\}, \{k,i\}\}$ and $v=\{i,j,k\}\backslash e$.
Note that for any cut, the contribution of the edges of a scenario to the weight of the cut is exactly equal to the load on the least loaded machine if we assign the jobs on one side of the cut to one machine and the jobs on the other side of the cut to the other machine.

We now show that the weights thus defined are non-negative. Substituting the expressions for $b_i,b_j,b_k$, we get that
\[w(i,j)=\tfrac12b_i+\tfrac12 b_j -\tfrac12 b_k = \tfrac12
\left(\min\{p_i, p_j+p_k\}+\min\{p_j,p_i+p_k\}-\min\{p_k,p_i+p_j\}\right).\]
Now, noting that either $\min\{p_i, p_j+p_k\}+\min\{p_j,p_i+p_k\} = p_i+p_j$, or $\min\{p_i, p_j+p_k\}+\min\{p_j,p_i+p_k\}\ge p_k$,
we get that $\min\{p_i, p_j+p_k\}+\min\{p_j,p_i+p_k\}\ge \min\{p_k,p_i+p_j\}$, and, thus, $w(i,j)\ge0$.


Let $L=\sum_{S\in {\cal S}} p(S)$, and for the optimal solution to SM2, let $L^*_{\min}$ be the sum over all scenarios of the load of the least loaded machine, and let $L^*_{\max}$ be the sum over all scenarios of the load of the most loaded machine, i.e., the sum of makespans. Then $L=L^*_{\min}+L^*_{\max}$. We denote by $ALG(L_{\min})$ and $ALG(L_{\max})$ the sums over all scenarios of the loads on the least loaded and most loaded machines defined by the cut.
A $(1-\gamma)$-approximation to {\sc Max Cut} gives us an assignment of jobs to machines such that 
$ALG(L_{\min}) \geq (1-\gamma)L^*_{\min}=(1-\gamma)(L-L^*_{\max})$. The sum of the makespans over all scenarios is
$ALG(L_{\max}) \leq L-(1-\gamma)(L-L^*_{\max})=\gamma L + (1-\gamma)L^*_{\max}$. Now, note that the makespan for any scenario is at least half of the sum of the processing times, and hence $L \leq 2 L^*_{\max}$. So, a $(1-\gamma)$-approximation for {\sc Max Cut} implies a $(1+\gamma)$-approximation for SM2 in the case where all scenarios have at most three jobs.
\end{proof}

\end{document}